\documentclass[10pt]{article}

\usepackage[T1]{fontenc}

\usepackage{amsmath}
\usepackage{amssymb}
\usepackage{amsthm}
\usepackage{tabularx}
\usepackage{tikz}
\usetikzlibrary{positioning, calc, shapes.geometric}

\newcommand{\probname}[1]{\textnormal{\textsc{#1}}}
\newcommand{\gm}[0]{\probname{Gerrymandering}}
\newcommand{\lp}[0]{\probname{$k$-Labeled Path}}
\newcommand{\scov}[0]{\probname{Set Cover}}
\newcommand{\wgm}[0]{\probname{Weighted Gerrymandering}}
\newcommand{\uwgm} [0] {\probname{Unit Weight Gerrymandering}}

\DeclareMathOperator*{\argmax} {arg\ max}

\newcommand{\D}{\mathcal{D}}
\newcommand{\Un}{\mathcal{U}}
\newcommand{\F}{\mathcal{F}}

\newcommand{\problembox} [3] {
    \vspace{\dimexpr\parskip+1.35ex}
    \noindent
    \begin{tikzpicture}
        \node[draw=black!40, rounded corners, inner sep=2ex] (content) {
            \begin{tabularx} {\dimexpr\columnwidth-4ex-0.3pt} {l X}
                \textbf{Input:} & #2\\
                \textbf{Problem:} & #3\\
            \end{tabularx}
        };

        \node[inner sep=3pt, fill=white, anchor=north west] at ($(content.north west) + (2ex, 1.35ex)$) {\probname{#1}};
    \end{tikzpicture}%
}

\definecolor{orange}{RGB}{230,159,0}
\definecolor{skyblue}{RGB}{86,180,233}
\definecolor{bluegreen}{RGB}{0,158,115}
\definecolor{yellow}{RGB}{240,228,66}
\definecolor{blue}{RGB}{0,114,178}
\definecolor{vermillion}{RGB}{213,94,0}
\definecolor{redpurple}{RGB}{204,121,167}

\newcommand{\bolddots} [0] {\Huge{\textbf{.\hspace{-0.8mm}.\hspace{-0.8mm}.}}}

\theoremstyle{plain}
\newtheorem{theorem} {Theorem} [section]
\newtheorem{lemma} [theorem] {Lemma}
\newtheorem{observation} [theorem] {Observation}
\newtheorem{corollary} [theorem] {Corollary}

\theoremstyle{definition}
\newtheorem{definition} [theorem] {Definition}

\title{Parameterized Complexity of Gerrymandering}
\author{Andrew Fraser, Brian Lavallee, and Blair D.\ Sullivan}
\date{December 2023}

\begin{document}

\maketitle

\let\thefootnote=1
\footnotetext{This work was supported in part by the Gordon \& Betty Moore Foundation under award GBMF4560 to Blair D.\ Sullivan.}

\begin{abstract}
    In a representative democracy, the electoral process involves partitioning geographical space into districts which each elect a single representative.
    These representatives craft and vote on legislation, incentivizing political parties to win as many districts as possible (ideally a plurality).
    Gerrymandering is the process by which district boundaries are manipulated to the advantage of a desired candidate or party.
    We study the parameterized complexity of \gm{}, a graph problem (as opposed to Euclidean space) formalized by Cohen-Zemach et al.\ (AAMAS 2018) and Ito et al.\ (AAMAS 2019) where districts partition vertices into connected subgraphs.
    We prove that \uwgm{} is W[2]-hard on trees (even when the depth is two) with respect to the number of districts $k$.
    Moreover, we show that \uwgm{} remains W[2]-hard in trees with $\ell$ leaves with respect to the combined parameter $k+\ell$.
    In contrast, Gupta et al.\ (SAGT 2021) give an FPT algorithm for \gm{} on paths with respect to $k$.
    To complement our results and fill this gap, we provide an algorithm to solve \gm{} that is FPT in $k$ when $\ell$ is a fixed constant.
\end{abstract}

\let\thefootnote=2
\section{Introduction}
Many electoral systems around the world divide voters into districts.
The votes in each district are tallied separately, and each district elects a representative to a seat in a congressional system.
The adversarial manipulation of these districts to favor one political party over another is known as \textit{gerrymandering} and has the potential to greatly skew elections.
Gerrymandering has been studied in various contexts, including political science~\cite{mcghee2020}, geography~\cite{lewenberg2017}, and social networks~\cite{Talmon2018,tsang2016}.
Many studies focus on the prevention of gerrymandering~\cite{chen2013} or calculate a fairness metric on real-world districts~\cite{chen2016,cottrell2019,ko2022,simon2020}.
The increasing role of algorithms in the creation and evaluation of district maps~\cite{tam2019,earle2018} motivates the study of the computational complexity of gerrymandering problems.

In this paper, we study gerrymandering in the graph setting.
Cohen-Zemach et al.\ proposed a model in which the vertices of a graph represent (groups of) voters and edges model proximity and continuity~\cite{cohen2018}.
Compared to a geographic map, this abstraction is more general and therefore more powerful, since a graph can represent complex socio-political relationships in addition to geographical proximity.
Ito et al.\ followed this notion and formally defined the \gm{} problem~\cite{ito2019}.
Given a graph, \gm{} asks if the vertices can be partitioned into connected subsets so that a preferred candidate (or political party) wins the most districts.
A candidate wins a district by receiving the most votes, represented by vertex weights, within a district.

Several hardness results have been shown for this problem.
In 2019, Ito et al.\ proved \gm{} is NP-complete even when restricted to complete bipartite graphs with only $k=2$ districts and 2 candidates~\cite{ito2019}.
They also observed a simple $O(n^k)$ algorithm for trees (proving that \gm{} is XP with respect to $k$) and gave a polynomial time algorithm for stars.
In 2021, Bentert et al.\ proved that even \uwgm{} remains NP-hard on paths~\cite{bentert2023}.
They also prove that \gm{} is weakly NP-hard on trees with 3 candidates, but it becomes solvable in polynomial time when there are only 2 candidates.

In 2021, Gupta et al.\ showed that \gm{} is fixed parameter tractable (FPT) on paths with respect to the number of districts $k$ (independent of the number of candidates)~\cite{gupta2021}.
They gave an $O(2.619^k (n+m)^{O(1)})$ algorithm for \wgm{}, a generalization of \gm{} which allows vertices to split their votes between multiple candidates.

In this paper, we study the parameterized complexity of \gm{} in trees.
We prove that \uwgm{} is W[2]-hard\footnote{The version of this paper in SAGT 2023 only proves these results for \gm{}.} on trees (even when the depth is 2) with respect to the number of districts $k$, suggesting that no FPT algorithm exists.
This contrasts sharply with the polynomial time algorithm for stars (trees of depth 1), and answers an open question of Gupta et al.~\cite{gupta2021}.
To better understand the difference in complexity between trees and paths, we also study the problem in trees with only $\ell$ leaves.
In this setting, we prove that \uwgm{} is still W[2]-hard\footnotemark[2] with respect to the combined parameter $k + \ell$, even on subdivided stars (i.e.\ when only one vertex has degree greater than 2).
To complement this result, we also provide an algorithm for \wgm{} in trees with $\ell$ leaves.
The algorithm is FPT with respect to $k$ when $\ell$ is a fixed constant.

\section{Preliminaries}
\gm{} from~\cite{ito2019} is defined on simple, undirected graphs.
We use $C$ to denote the set of all candidates, and we annotate a graph $G = (V, E)$ with each vertex $v$ having a candidate preference $\chi(v)$ and number of votes cast $w(v)$.
Given a graph, \gm{} asks for a \emph{district-partition} of $G$ with $k$ districts.

\vspace{1em}

\begin{definition}
	Given a graph $G = (V, E)$, a \textbf{district-partition} of $G$ is a partition of $V$ into sets $D_1, \dots, D_k$ so that $D_1 \cup \dots \cup D_k = V$, $D_i \cap D_j = \emptyset$ for all $i \neq j$, and the induced subgraph $G[D_i]$ is connected for all $i$.
	We refer to $D_1, \dots, D_k$ as \textbf{districts}.
\end{definition}

Specifically, \gm{} asks for a district-partition in which a preferred candidate $p$ (equivalently a group of affiliated candidates or a political party) \emph{wins} a plurality of districts.
We refer to the following definition from Ito et al.\ \cite{ito2019}.

\begin{definition} \label{def:wins}
	Given a graph $G = (V, E)$ and a district-partition $\D$ of $G$, we define the set of all candidates with the most votes in a district $D \in \D$ as follows:
	\[\text{top}(D) := \argmax_{q \in C} \left\{\sum_{v \in D : \chi(v) = q} w(v) \right\}\]
\end{definition}

We say that a candidate $q$ \emph{leads} a district $D$ if $q \in \text{top}(D)$.
We note that Definition~\ref{def:wins} allows for multiple candidates to lead a single district.
If $\text{top(D)} = \{q\}$ and therefore $q$ is the only leader of the district, we say that $q$ \emph{wins} the district.

Formally, \gm{} is defined as follows.

\problembox{Gerrymandering}
{A graph $G = (V,E)$, a set of candidates $C$, a candidate function $\chi: V \rightarrow C$, a weight function $w: V \rightarrow \mathbb{N}$, a preferred candidate $p \in C$, and an integer $k \in \mathbb{N}$.}
{Is there a district-partition of $V$ into $k$ districts such that $p$ wins more districts than any other candidate leads?}

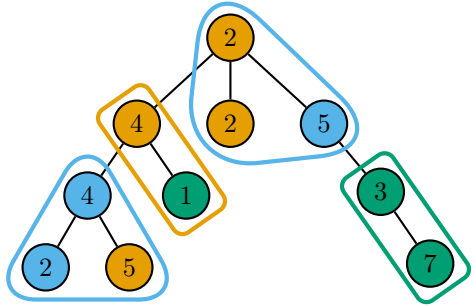
\begin{figure}[t]
	\begin{center}
		\begin{tikzpicture}
	\tikzstyle{node} = [circle, draw, thick]
	\tikzstyle{edge} = [thick]

	\node (1) [node, fill=orange] {4};
	\node (2) [node, fill=orange, right=0.6cm of 1] {2};
	\node (3) [node, fill=skyblue, right=0.6cm of 2] {5};
	\node (4) [node, fill=orange, above=0.5cm of 2] {2};

	\node (5) [node, fill=bluegreen, below right=0.45cm and 0.3cm of 3] {3};
	\node (6) [node, fill=bluegreen, below right=0.5cm and 0.2cm of 1] {1};
	\node (12) [node, fill=bluegreen, below right=0.5cm and 0.2cm of 5] {7};

	\node (9) [node, fill=skyblue, below left=0.5cm and 0.2cm of 1] {4};
	\node (10) [node, fill=skyblue, below left=0.5cm and 0.1cm of 9] {2};
	\node (11) [node, fill=orange, below right=0.5cm and 0.1cm of 9] {5};

	\draw (1) edge [edge] (4);
	\draw (1) edge [edge] (6);
	\draw (1) edge [edge] (9);
	\draw (2) edge [edge] (4);
	\draw (3) edge [edge] (5);
	\draw (3) edge [edge] (4);

	\draw (9) edge [edge] (10);
	\draw (11) edge [edge] (9);
	\draw (12) edge [edge] (5);

	\draw[ultra thick, rounded corners=6mm, color=skyblue] ($(9.north)+(0,0.5)$) -- ($(10.south west)+(-0.5,-0.2)$) -- ($(11.south east)+(0.5,-0.2)$) -- cycle;

	\draw[ultra thick, rotate=-55, rounded corners, color=orange] ($(1.north west)+(-0.1,0.45)$) rectangle ($(6.south east)+(0.1, -0.45)$);

	\draw[ultra thick, rounded corners=10mm, color=skyblue] ($(4.north west)+(-0.1,0.7)$) -- ($(2.south west)+(-0.4,-0.4)$) -- ($(3.south east)+(0.7,-0.1)$) -- cycle;

	\draw[ultra thick, rotate=-55, rounded corners, color=bluegreen] ($(5.north west)+(-0.1, 0.45)$) rectangle ($(12.south east)+(0.1, -0.45)$);

\end{tikzpicture}
	\end{center}
	\caption{
		A district-partition for a \gm{} instance with $k=4$ districts and $|C| = 3$ candidates.
		Each vertex $v$ is labeled with $w(v)$, the number of votes cast by $v$.
		Vertices are colored according to which candidate $\chi(v)$ they vote for.
		Blue indicates the preferred candidate $p$.
		Each district is outlined in the color of its winning candidate.
		Candidate $p$ wins two districts and the other two candidates each win only one, so this is a satisfying district-partition.
	}
	\label{fig:gerrymandering}
\end{figure}

Figure~\ref{fig:gerrymandering} shows a small example.
\uwgm{} is the natural restriction of \gm{} in which $w(v) = 1$ for all $v \in G$.
We defer the definition of \wgm{} to Section~\ref{sec:xp_alg} to avoid notational conflicts.
Our work focuses on the parameterized complexity of \gm{} on trees.
The study of parameterized complexity revolves around two important classes of problems.

\begin{definition} \label{def:FPT}
    A problem $\Pi$ is \textbf{fixed parameter tractable} (\textbf{FPT}) with respect to a parameter $k$
	if it admits an algorithm $A$ which can answer an instance of $\Pi$ of size $n$ with parameter value $k$ in time $O(f(k) \cdot n^{O(1)})$ for some computable function $f$.
	We call $A$ an FPT algorithm.
	Similarly, $\Pi$ is \textbf{slicewise polynomial} (\textbf{XP}) with respect to $k$ if $A$ runs in time $O(g(k) \cdot n^{h(k)})$ for computable functions $g,h$.
	In this case, we call $A$ an XP algorithm.
\end{definition}

Clearly, FPT $\subseteq$ XP, and so the study of parameterized complexity often focuses on determining whether or not a problem is in FPT.
Like P $\neq$ NP, FPT $\neq$ W[1] is the basic complexity assumption at the foundation of parameterized algorithms.
W[1]-hardness is proven using parameterized reductions which resemble standard NP-hardness reductions but have additional requirements on the translation of the parameter.
In this paper, we prove W[2]-hardness (an even stronger notion~\cite{cygan2015}) via parameterized reductions from \scov{}.

\problembox{Set Cover}
{A set of elements $\Un = \{e_1, \dots, e_n\}$, a family of sets $\F = \{S_1, \dots, S_m\}$, and an integer $t \in \mathbb{N}$.}
{Is there a subset $X \subseteq \F$ such that $|X| \leq t$ and $\bigcup_{S \in X} S = \Un$?}

\scov{} is a well-studied problem in the field of parameterized complexity which is known to be W[2]-hard when parameterized by the natural parameter $t$~\cite{flum2006}.
We assume that every element in $\Un$ appears in at least one set of $\F$, as otherwise it is trivially a NO-instance.
Moreover, we assume that $t \leq |\F|$ since otherwise it is a trivial YES-instance.
We refer to the textbook by Cygan et al.~\cite{cygan2015} for additional reading on parameterized complexity.

\section{W[2]-Hardness in Trees of Depth Two}
In this section, we prove that the \uwgm{} problem is W[2]-hard in trees of depth 2 parameterized by the number of districts $k$ using a reduction from \scov{}.
Before describing the reduction, we make the following observation about the frequency of elements in an instance of \scov{}.

\vspace{1em}

\begin{observation} \label{obs:sc_freq}
	Let $(\Un, \F, t)$ be an instance of \scov{}, and let $f_e$ denote the frequency of $e$: the number of sets in $\F$ which contain the element $e \in \Un$.
	There is an equivalent instance $(\Un, \F', t)$ in which $f_e = f_{e'}$ for all $e, e' \in \Un$.
\end{observation}

The equivalent instance can be constructed by adding additional sets to $\F$ which contain only a single element $e$.
This increases the frequency of $e$ by one and can be repeated until $f_e$ is equal to the maximum frequency of any element.
Since the new sets contain only a single element, they can be replaced in a feasible solution by any set in $\F$ which also contains that element.

\begin{figure}[t]
	\begin{center}
		\resizebox{0.75\linewidth}{!}{
			\begin{tikzpicture}
	\tikzstyle{node} = [circle, draw, thick, minimum size=8mm, inner sep=0]

	\node (r) [node, fill=skyblue] {$r$};

	\node (w1) [node, fill=skyblue, above left=8mm and 6mm of r] {$w_1$};
	\node (w1p) [node, fill=orange, above left=6mm and -2mm of w1] {$w'_1$};
	\node (wdots) [right=2mm of w1] {\huge{$\dots$}};
	\node (wd) [node, fill=skyblue, above right=8mm and 6mm of r] {$w_d$};
	\node (wdp) [node, fill=orange, above right=6mm and -2mm of wd] {$w'_d$};

	\draw (r) edge [thick] (w1);
	\draw (w1) edge [thick] (w1p);
	\draw (r) edge [thick] (wd);
	\draw (wd) edge [thick] (wdp);

	\draw[ultra thick, rounded corners=3mm, color=skyblue] ($(w1.north west) + (-0.25, 0.25)$) rectangle ($(w1.south east) + (0.25, -0.25)$);
	\draw[ultra thick, rounded corners=3mm, color=orange] ($(w1p.north west) + (-0.25, 0.25)$) rectangle ($(w1p.south east) + (0.25, -0.25)$);

	\node (s1) [node, fill=bluegreen, below left=10mm and 20mm of r] {$s_1$};
	\node (s1p) [node, fill=bluegreen, below left=-0.425mm and 9.337mm of s1] {$s'_1$};
	\node (11) [node, fill=blue, below left=9mm and 3mm of s1] {$v_1^1$};
	\node (21) [node, fill=redpurple, below right=10mm and -3mm of s1] {$v_1^2$};

	\draw (r) edge [thick] (s1);
	\draw (s1) edge [thick] (s1p);
	\draw (s1) edge [thick] (11);
	\draw (s1) edge [thick] (21);

	\draw[ultra thick, rounded corners=6mm, color=bluegreen] ($(21.south east) + (0.25, -0.25)$) rectangle ($(s1p.north west) + (-0.25, 0.75)$);

	\node (s2) [node, fill=yellow, below left=10mm and -2mm of r] {$s_2$};
	\node (s2p) [node, fill=yellow, below left=10mm and 0mm of s2] {$s'_2$};
	\node (22) [node, fill=redpurple, below right=10mm and 0mm of s2] {$v_2^2$};

	\draw (r) edge [thick] (s2);
	\draw (s2) edge [thick] (s2p);
	\draw (s2) edge [thick] (22);

	\node (sdots) [right=5mm of s2] {\huge{$\dots$}};
	\node (sm) [node, fill=vermillion, below right=10mm and 20mm of r] {$s_m$};
	\node (smp) [node, fill=vermillion, below left=10mm and -3mm of sm] {$s'_m$};
	\node (1d) [node, fill=blue, below right=9mm and 3mm of sm] {$v_d^1$};
	\node (3d) [node, fill=white, below right=-0.425mm and 9.337mm of sm] {$v_d^3$};

	\draw (r) edge [thick] (sm);
	\draw (sm) edge [thick] (smp);
	\draw (sm) edge [thick] (1d);
	\draw (sm) edge [thick] (3d);

	\draw[ultra thick, rounded corners=6mm, color=vermillion] ($(3d.north east) + (0.25, 0.75)$) rectangle ($(smp.south west) + (-0.25, -0.25)$);
\end{tikzpicture}
		}
	\end{center}
	\caption{
		An instance of \uwgm{} constructed from a \scov{} instance where $S_1 = \{e_1, e_2\}$, $S_2 = \{e_2\}$, and $S_m = \{e_1, e_3\}$.
		The vertex labels match those used in Definition~\ref{def:depth_hardness}.
		Vertices are colored according to which candidate $\chi(v)$ they vote for.
		Part of a satisfying district-partition corresponding to a set cover including $S_1$ and $S_m$ is shown by the colored borders.
		Since each element candidate has a vertex in a district won by a set candidate, the preferred candidate wins the root district.
	}
	\label{fig:depth_hardness}
\end{figure}
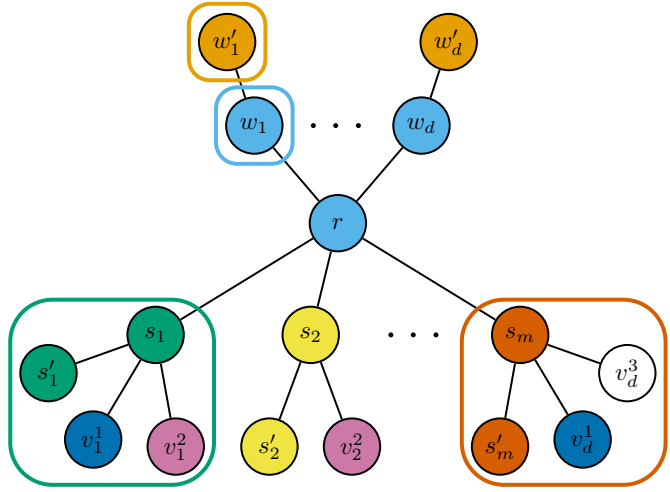

Now, we describe how to construct an instance of \uwgm{} from an instance of \scov{} (see Figure~\ref{fig:depth_hardness}).

\begin{definition} \label{def:depth_hardness}
	Let $(\Un, \F, t)$ be an instance of \scov{}.
	By Observation~\ref{obs:sc_freq}, we may assume that $f_e = d$ for all $e \in \Un$.
	Note that we will assume $d \geq 3$.
	We construct an instance $(G, C, \chi, p, k)$ of \uwgm{} as follows.
	Set $|C| = n+m+2$ and $k = t+3$.
	The \emph{element candidates} $a_1, \dots, a_n$ correspond to $e_1, \dots, e_n \in \Un$, and the \emph{set candidates} $b_1, \dots, b_m$ correspond to $S_1, \dots, S_m \in \F$.
	Finally, $p$ is the \emph{preferred candidate}, and $q$ is the \emph{adversary}.

	We construct the tree $G$ starting from its root vertex $r$; set $\chi(r) = p$.
	Add $d$ \emph{weight branches} to $r$ each consisting of 2 vertices, $w_i$ and $w'_i$, such that $w_i$ is adjacent to $r$.
	Set $\chi(w_i) = p$ and $\chi(w'_i) = q$ on each branch.
	For each set $S_i \in \F$, add another branch to $r$ with vertices $s_i$ and $s'_i$, and set $\chi(s_i) = \chi(s'_i) = b_i$.
	For each element $e_j$, add $d$ vertices $v_1^j, \dots, v_d^j$ to $G$, and connect one to each $s_i$ such that $e_j \in S_i$.
	Set $\chi(v_\ell^j) = a_j$.
\end{definition}

The construction of Definition~\ref{def:depth_hardness} works in two steps.
First, the adversary $q$ forces $p$ to win the root district.
Any district won by $p$ on a weight branch must be paired with another won by $q$, and so $p$ can only get ahead of $q$ by winning the root district.
Then, in order for $p$ to win the root district, at least one of the $d$ vertices voting for the element candidate $a_j$ must be in a separate district.
Since we must do this for every element with only $t$ districts, the only efficient solution corresponds to a set cover.
We begin by formally stating the first step.

\begin{observation} \label{obs:win_root}
	Let $(G, C, \chi, p, k)$ be an instance of \uwgm{} produced from an instance of \scov{} $(\Un, \F, t)$ according to Definition~\ref{def:depth_hardness}.
	Given a district-partition $\D$, if the preferred candidate $p$ does not win the root district (i.e.\ the district containing $r$), then $\D$ does not witness that $(G, C, \chi, p, k)$ is a YES-instance.
\end{observation}

Observation~\ref{obs:win_root} follows from the fact that the only way to create a district in $G$ which does not contain $r$ but is still won by $p$ is to place $w_i$ alone in a district.
However, this necessitates a second district containing only $w'_i$, and thus the adversary $q$ must win as many districts as $p$.

\begin{theorem} \label{thm:gerry_w2}
	\uwgm{} is W[2]-hard in trees of depth 2 when parameterized by the number of districts $k$.
\end{theorem}

\begin{proof}
	Let $(\Un, \F, t)$ be an instance of \scov{}.
	We will prove the claim by showing that $(\Un, \F, t)$ is equivalent to the instance $(G, C, \chi, p, k)$ of \uwgm{} constructed according to Definition~\ref{def:depth_hardness}.
	First, we prove that a YES-instance of \scov{} produces a YES-instance of \uwgm{}.

	Let $X \subseteq \F$ be a set cover of size $t$ witnessing that $(\Un, \F, t)$ is a YES-instance.
	Note that we may assume $|X| = t$ since adding sets to a feasible cover cannot make it infeasible.
	Let $\D$ be the following district-partition of $G$.
	For each $S_i \in X$, create a district containing $s_i$ and all of its children; candidate $b_i$ wins the district created for $S_i$.
	Create one district containing only $w_1$ and another containing only $w'_1$; candidate $p$ and the adversary $q$ win these districts respectively.
	Place all remaining vertices (including the root) in the final district.

	Note that $\D$ contains exactly $k = t + 3$ districts, each of which is connected.
	We need only show that candidate $p$ wins the root district.
	Suppose not.
	There must exist a candidate which receives at least $d \geq 3$ votes since $p$ receives $d$ votes from $r$ and $w_2, \dots, w_d$.
	The set candidates can only receive two votes in the entire instance, and so it must be an element candidate.
	Without loss of generality, suppose it is candidate $a_j$ corresponding to element $e_j$.
	In order to receive $d$ votes, all $d$ element-vertices $v_1^j, \dots, v_d^j$ must be in the root district.
	The construction of $\D$ thus implies that $e_j \not\in \bigcup_{S_i \in X} S_i$, contradicting that $X$ is a set cover.
	Therefore, $p$ wins the root district.

	Now, we show that a YES-instance of \uwgm{} implies a YES-instance of \scov{}.
	Let $\D$ be a satisfying district-partition of $G$.
	Note that $\D$ must contain at least two distinct districts with a vertex voting for $p$.
	By construction, one of these districts is a subset of a single weight branch.
	Furthermore, this weight branch must actually contain two districts, since $p$ only leads a district containing both $w_i$ and $w'_i$.
	Let $X \subseteq F$ be the set containing each $S_i$ such that a vertex in the subtree rooted at $s_i$ does not appear in the root district of $\D$.
	Since two districts in $\D$ appear on a weight branch and non-root districts can only intersect the subtree of a single $s_i$, $|X| \leq t$.

	Suppose that $X$ is not a set cover of $(\Un, \F, t)$.
	Then there exists some element $e_j$ which is not contained by any set in $X$.
	By the construction of $X$, this implies that all of $v_1^j, \dots, v_d^j$ appear in the root district of $\D$.
	However, then candidate $a_j$ would receive $d$ votes in the root district, and so candidate $p$ could not have won (since it receives at most $d$ votes as well).
	This contradicts that $\D$ was a satisfying district-partition, since $p$ must win the root district in order to win a plurality of districts by Observation~\ref{obs:win_root}.
\end{proof}

\section{W[2]-Hardness in Trees with Few Leaves}

In this section, we prove that \uwgm{} is W[2]-hard on subdivided stars parameterized by the combined parameter of $k+\ell$, where $\ell$ is the number of leaves.
We again reduce from the \scov{} problem.

The general idea of our reduction is to create a subdivided star with $t$ main branches, each containing groups of vertices representing every set in $\F$.
Choosing some group on a branch as an endpoint of the root district corresponds to choosing a set to add to the cover.
The hardness of the problem lies in the many choices of endpoints for the root district on each branch.
If the root district chooses these boundaries such that the corresponding sets are a set cover, $p$ will win more districts than any other candidate.

To begin our reduction, we describe how to produce an instance of \uwgm{} from an instance of \scov{}.
Figure~\ref{fig:leaf_hardness} demonstrates an example of the construction.

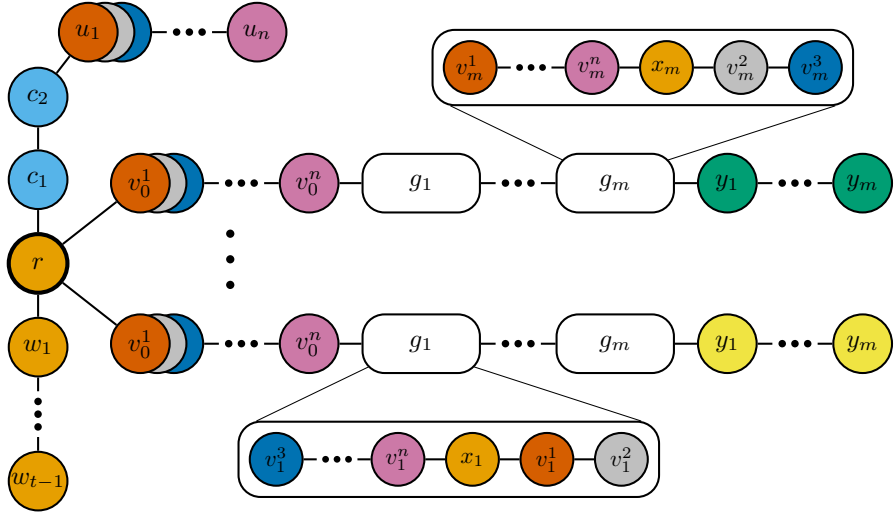
\begin{figure}[t]
    \resizebox{\textwidth}{!}{
        \begin{tikzpicture}
    \tikzstyle{node} = [circle, draw, thick, minimum size=8mm, inner sep=0pt]
    \tikzstyle{gnode} = [rectangle, draw, thick, minimum width=16mm, minimum height=8mm, rounded corners=3mm]

    \node (r) [node, fill=orange, ultra thick] {$r$};

    \node (c1) [node, fill=skyblue, above=3mm of r] {$c_1$};
    \node (c2) [node, fill=skyblue, above=3mm of c1] {$c_2$};

    \node (un) [node, fill=redpurple, above right=3mm and 24mm of c2] {$u_n$};
    \node (ud) [left=2mm of un, inner sep=0pt] {\bolddots};
    \node (u3) [node, fill=blue, left=2mm of ud] {$u_3$};
    \node (u2) [node, fill=lightgray, left=-6mm of u3] {$u_2$};
    \node (u1) [node, fill=vermillion, left=-6mm of u2] {$u_1$};

    \draw (r) edge [thick] (c1);
    \draw (c1) edge [thick] (c2);
    \draw (c2) edge [thick] (u1);
    \draw (u3) edge [thick] (ud);
    \draw (ud) edge [thick] (un);

    \node (w1) [node, fill=orange, below=3mm of r] {$w_1$};
    \node (wd) [circle, below=5mm of w1, anchor=center, inner sep=0pt, rotate=90] {\bolddots};
    \node (wt) [node, fill=orange, below=5mm of wd.center] {$w_{t-1}$};

    \draw (r) edge [thick] (w1);
    \draw (w1) edge [thick] (wd);
    \draw (wd) edge [thick] (wt);

    \node (v0n1) [node, fill=redpurple, above right=5mm and 31mm of r] {$v_0^n$};
    \node (v0d1) [left=2mm of v0n1, inner sep=0pt] {\bolddots};
    \node (v031) [node, fill=blue, left=2mm of v0d1] {$v_0^3$};
    \node (v021) [node, fill=lightgray, left=-6mm of v031] {$v_0^2$};
    \node (v011) [node, fill=vermillion, left=-6mm of v021] {$v_0^1$};

    \node (g11) [gnode, right=3mm of v0n1] {$g_1$};
    \node (gd1) [right=2mm of g11, inner sep=0pt] {\bolddots};
    \node (gm1) [gnode, right=2mm of gd1] {$g_m$};

    \node (y11) [node, fill=bluegreen, right=3mm of gm1] {$y_1$};
    \node (yd1) [right=2mm of y11, inner sep=0pt] {\bolddots};
    \node (ym1) [node, fill=bluegreen, right=2mm of yd1] {$y_m$};

    \draw (r) edge [thick] (v011);
    \draw (v031) edge [thick] (v0d1);
    \draw (v0d1) edge [thick] (v0n1);
    \draw (v0n1) edge [thick] (g11);
    \draw (g11) edge [thick] (gd1);
    \draw (gd1) edge [thick] (gm1);
    \draw (gm1) edge [thick] (y11);
    \draw (y11) edge [thick] (yd1);
    \draw (yd1) edge [thick] (ym1);

    \node (dots) [below right=-5mm and 0mm of v031, scale=2.5] {\Large{$\vdots$}};

    \node (vm1) [node, scale=0.9, fill=vermillion, above left=9mm and 9mm of gm1] {$v_m^1$};
    \node (vmd) [scale=0.9, right=1.8mm of vm1, inner sep=0pt] {\bolddots};
    \node (vmn) [node, scale=0.9, fill=redpurple, right=1.8mm of vmd] {$v_m^n$};
    \node (xm) [node, scale=0.9, fill=orange, right=2.7mm of vmn] {$x_m$};
    \node (vm2) [node, scale=0.9, fill=lightgray, right=2.7mm of xm] {$v_m^2$};
    \node (vm3) [node, scale=0.9, fill=blue, right=2.7mm of vm2] {$v_m^3$};

    \draw (vm1) edge [thick] (vmd);
    \draw (vmd) edge [thick] (vmn);
    \draw (vmn) edge [thick] (xm);
    \draw (xm) edge [thick] (vm2);
    \draw (vm2) edge [thick] (vm3);

    \draw ($(gm1.north west) + (0.1, -0.1)$) -- ($(vm1.south) + (-0.27, -0.15)$);
    \draw ($(gm1.north east) + (-0.1, -0.1)$) -- ($(vm3.south) + (0.27, -0.15)$);
    \draw[thick, rounded corners=3mm] ($(vm1.north west) + (-0.25, 0.25)$) rectangle ($(vm3.south east) + (0.25, -0.25)$);

    \node (v0n2) [node, fill=redpurple, below right=5mm and 31mm of r] {$v_0^n$};
    \node (v0d2) [left=2mm of v0n2, inner sep=0pt] {\bolddots};
    \node (v032) [node, fill=blue, left=2mm of v0d2] {$v_0^3$};
    \node (v022) [node, fill=lightgray, left=-6mm of v032] {$v_0^2$};
    \node (v012) [node, fill=vermillion, left=-6mm of v022] {$v_0^1$};

    \node (g12) [gnode, right=3mm of v0n2] {$g_1$};
    \node (gd2) [right=2mm of g12, inner sep=0pt] {\bolddots};
    \node (gm2) [gnode, right=2mm of gd2] {$g_m$};

    \node (y12) [node, fill=yellow, right=3mm of gm2] {$y_1$};
    \node (yd2) [right=2mm of y12, inner sep=0pt] {\bolddots};
    \node (ym2) [node, fill=yellow, right=2mm of yd2] {$y_m$};

    \draw (r) edge [thick] (v012);
    \draw (v032) edge [thick] (v0d2);
    \draw (v0d2) edge [thick] (v0n2);
    \draw (v0n2) edge [thick] (g12);
    \draw (g12) edge [thick] (gd2);
    \draw (gd2) edge [thick] (gm2);
    \draw (gm2) edge [thick] (y12);
    \draw (y12) edge [thick] (yd2);
    \draw (yd2) edge [thick] (ym2);

    \node (v13) [node, scale=0.9, fill=blue, below left=9mm and 9mm of g12] {$v_1^3$};
    \node (v1d) [scale=0.9, right=1.8mm of v13, inner sep=0pt] {\bolddots};
    \node (v1n) [node, scale=0.9, fill=redpurple, right=1.8mm of v1d] {$v_1^n$};
    \node (x1) [node, scale=0.9, fill=orange, right=2.7mm of v1n] {$x_1$};
    \node (v11) [node, scale=0.9, fill=vermillion, right=2.7mm of x1] {$v_1^1$};
    \node (v12) [node, scale=0.9, fill=lightgray, right=2.7mm of v11] {$v_1^2$};

    \draw (v13) edge [thick] (v1d);
    \draw (v1d) edge [thick] (v1n);
    \draw (v1n) edge [thick] (x1);
    \draw (x1) edge [thick] (v11);
    \draw (v11) edge [thick] (v12);

    \draw ($(g12.south west) + (0.1, 0.1)$) -- ($(v13.north) + (-0.27, 0.15)$);
    \draw ($(g12.south east) + (-0.1, 0.1)$) -- ($(v12.north) + (0.27, 0.15)$);
    \draw[thick, rounded corners=3mm] ($(v13.north west) + (-0.25, 0.25)$) rectangle ($(v12.south east) + (0.25, -0.25)$);
\end{tikzpicture}
    }
    \caption{
        A \uwgm{} instance constructed from a \scov{} instance according to Definition~\ref{def:unsplit}.
        Each vertex $v$ is colored according to which candidate $\chi(v)$ it votes for.
        Specifically, light blue indicates $p$ and orange indicates $q$.
        Each set $S_i$ is represented by the group $g_i$ of $n+1$ vertices.
        In this example, $S_1 = \{e_1, e_2\}$ and $S_m = \{e_2, e_3\}$.
        Thus, $v_m^2$ and $v_m^3$ appear after $x_m$ in group $g_m$ (shown in the top branch), and $v_1^1$ and $v_1^2$ appear after $x_1$ (shown in the bottom branch).
    }
    \label{fig:leaf_hardness}
\end{figure}

\begin{definition} \label{def:unsplit}
    Given an instance $(\Un, \F, t)$ of \scov{}, let $(G, C, \chi, p, k)$ be the following instance of \uwgm{}.
    Set $|C| = n+t+2$ and $k = t+4$.
    The \emph{element candidates} $a_1, \dots, a_n$ correspond to $e_1, \dots, e_n \in \Un$.
    The \emph{branch candidates} $b_1, \dots, b_t$ each appear on their own branch.
    Finally, $p$ is the \emph{preferred candidate}, and $q$ is the \emph{ally}.

    Let $G$ be the following subdivided star with root vertex $r$ and $\ell = t+2$ branches.
    Set $\chi(r) = q$.
    Starting at the root, the first branch of $G$ consists of $c_1$ and $c_2$ followed by $u_1, \dots, u_n$.
    Set $\chi(c_1) = \chi(c_2) = p$, and set $\chi(u_j) = a_j$.
    The second branch contains $w_1, \dots, w_{t-1}$ with $\chi(w_i) = q$.
    The other $t$ (nearly identical) \emph{selection branches} consist of the following vertices:
    \begin{itemize}
        \item
        $x_1, \dots, x_m$ such that $\chi(x_i) = q$

        \item
        $v_0^j, \dots, v_m^j$ such that $\chi(v_i^j) = a_j$ for each element $e_j \in \Un$

        \item
        $y_0, \dots, y_m$ such that $\chi(y_i) = b_s$ on branch $s \in [t]$
    \end{itemize}

    The vertices $v_0^1, \dots, v_0^n$ begin the branch, and $y_0, \dots, y_m$ end the branch.
    The remaining vertices are organized into $m$ groups $g_1, \dots, g_m$ such that $g_i$ contains $x_i$ and $v_i^1, \dots, v_i^n$.
    Within a group, $v_i^j$ appears before $x_i$ if and only if $e_j \not\in S_i$.
    Thus, exactly $i$ vertices before $x_i$ vote for candidate $a_j$ if element $e_j \in S_i$, and $i+1$ otherwise.
    This same construction is repeated $t$ times with only $\chi(y_i)$ differing between branches.
    Note that $G$ contains $O(tnm)$ vertices which is polynomial in the \scov{} instance size.
\end{definition}

To prove that Definition~\ref{def:unsplit} produces equivalent instances, we first observe that the vertices $c_1$ and $c_2$ must be placed in two districts of size one.
This follows from the fact that they are the only two vertices in $G$ which vote for $p$.
Since $k > 1$, $p$ must win at least two districts in order to win a plurality.
This leads to the following observation about the root district.

\begin{observation} \label{obs:unsplit}
    Let $(G, C, \chi, p, k)$ be an instance of \uwgm{} produced from an instance $(\Un, \F, t)$ of \scov{} according to Definition~\ref{def:unsplit}.
    Given a district-partition $\D$, if the ally $q$ does not win the root district (i.e.\ the district containing $r$), then $\D$ does not witness that $(G, C, \chi, p, k)$ is a YES-instance.
\end{observation}

Since $c_1$ and $c_2$ must appear in their own districts, each element candidate leads a district on the first branch, regardless of how the vertices are divided in $\D$.
As a result, an element candidate cannot lead the root district in a satisfying district-partition.
Since branch candidates do not have enough weight to even tie element candidates, this implies that $q$ must win the root district.

\begin{theorem} \label{thm:unsplit}
    \uwgm{} is W[2]-hard in subdivided stars with $\ell$ leaves when parameterized by the combined parameter $k+\ell$.
\end{theorem}

\begin{proof}
    Let $(\Un, \F, t)$ be an instance of \scov{}.
    We will prove the claim by showing that $(\Un, \F, t)$ is equivalent to the instance $(G, C, \chi, p, k)$ of \uwgm{} constructed according to Definition~\ref{def:unsplit}.
    First, we will prove that a YES-instance of \scov{} produces a YES-instance of \uwgm{}.

    Let $X \subseteq \F$ be a set cover of size $t$.
    Note that we may assume $|X| = t$ since adding sets to a feasible cover cannot make it infeasible.
    Let $\D$ be the following district-partition of $G$.
    Create a district containing only $c_1$, a district containing only $c_2$, and a district containing $u_1, \dots, u_n$.
    Thus, $p$ wins two districts, and each element candidate leads one.
    For each set $S_i \in X$, create a district on one of the selection branches containing all of the vertices after $x_i$ on the branch.
    Each of these districts are won by a branch candidate.
    Place the remaining vertices in the root district.

    Note that $\D$ contains exactly $k=t+4$ districts, each of which is connected.
    We need only show that $q$ wins the root district.
    Suppose not; then, there exists an element candidate $a_j$ which receives at least as many votes as $q$.
    Between the root and the second branch, $q$ begins with $t$ more votes than $a_j$.
    By the construction of the selection branches, $a_j$ receives one additional vote relative to $q$ on each branch where $e_j \not\in S_i$.
    In order to make up all $t$ votes, $e_j \not\in S_i$ for all $S_i \in X$, contradicting that $X$ is a set cover.
    Thus, $q$ wins the root district, $p$ is the only candidate to win or lead more than one district, and $(G, C, \chi, p, k)$ is a YES-instance.

    Now, we show that a YES-instance of \uwgm{} implies a YES-instance of \scov{}.
    Let $\D$ be a satisfying district-partition of $G$.
    Construct the corresponding set cover $X$ in the following manner.
    For each selection branch, add $S_i$ to $X$ if $x_i$ is the last vertex from $x_1, \dots, x_m$ to appear in the root district.
    Note that $|X| \leq t$ since there are only $t$ selection branches.

    Suppose that $X$ is not a feasible set cover of $(\Un, \F, t)$.
    Then, there exists an element $e_j$ such that $e_j \not\in S_i$ for all $S_i \in X$.
    This implies that candidate $a_j$ receives at least one more vote than $q$ on every selection branch.
    Thus, $q$ cannot win the root district since it only has weight $t$ outside the selection branches.
    This is a contradiction since $q$ must win the root district in a satisfying district-partition by Observation~\ref{obs:unsplit}.
    Therefore, $X$ must be a feasible set cover.
\end{proof}

\section{FPT Algorithm w.r.t.\ $k$ in Trees with Few Leaves}
\label{sec:xp_alg}

In this section, we provide an algorithm to solve \wgm{} that is FPT with respect to the number of districts $k$ and XP with respect to the number of leaves $\ell$.
\wgm{} is a generalization of \gm{} which allows a vertex to vote for multiple candidates.
In this setting, $C$ denotes the set of candidates, and $w(v)$ is a $|C|$-dimensional vector whose $i$-th component is the number of votes for candidate $i$.
The definitions of \emph{wins} and \emph{leads} are analogous to \gm{}.

\problembox{Weighted Gerrymandering}
{A graph $G = (V,E)$, a set of candidates $C$, a weight function $w: V \rightarrow C \times \mathbb{N}$, a preferred candidate $p$, and an integer $k \in \mathbb{N}$.}
{Is there a district-partition of $V$ into $k$ districts such that $p$ wins more districts than any other candidate leads?}

We use a modified version of the FPT algorithm for paths from Gupta et al.~\cite{gupta2021}.
Their algorithm creates an instance $(H, s, t, k, k^*)$ of \lp{} which, given a partially edge-labeled directed graph $H$ with vertices $s$ and $t$, asks if there exists an $st$-path with exactly $k$ internal vertices such that no edge label is used more than $k^*$ times.
They construct this instance by creating one vertex for each of the $\binom{n}{2}$ possible districts in $G$, connecting vertices in $H$ corresponding to adjacent districts in $G$, and labeling those edges by the winner of the preceding district (unless $p$ wins).
In this way, a satisfying district-partition of $G$ corresponds to a satisfying $st$-path in $H$ and vice versa.
We slightly modify their construction to solve on a collection of disjoint paths.

Note that the algorithm given by Gupta et al.~\cite{gupta2021} requires a tie-breaking rule $\eta$ as a parameter to the problem.
For any district $D \in V(G)$, the rule $\eta$ must declare a distinct winner from the set $\argmax_{q \in C}\{ \sum_{v \in D} w(v)[q]\}$.
Their algorithm applies this rule to ensure that no district has more than one winner.
The details in our algorithm don't directly apply this rule, as all district decisions are made using Corollary ~\ref{cor:disjoint_fpt_paths}.

\begin{corollary} \label{cor:disjoint_fpt_paths}
    Let $(G, C, \chi, w, p, k)$ be an instance of \wgm{} with a tie-breaking rule $\eta$.
    If $G$ is a path forest (one or more disconnected paths), then there exists an algorithm which can decide $(G, C, \chi, w, p, k)$ in time $O(2.619^k (n + m)^{O(1)})$.
\end{corollary}

\begin{proof}
    We construct a similar instance of \lp{} as in~\cite{gupta2021}.
    First, number the vertices of $G$ such that $[1, a]$ are in the first path, then $[a+1, b]$ are in the next path, and so forth.
    For any of the $\binom{n}{2}$ ``district'' vertices that would usually be created,
    only create the vertex if the endpoints $i$ and $j$ are both contained in the same path in $G$.
    We add labeled edges between these vertices in the same manner as~\cite{gupta2021}.
    By only creating vertices corresponding to legal districts of $G$, feasible solutions to the \lp{} instance must respect the disconnected structure of $G$.
    The remainder of the argument follows~\cite{gupta2021}.
\end{proof}

\begin{figure}[t]
    \begin{center}
        \begin{tikzpicture}
    \tikzstyle{node} = [circle, draw, thick, minimum size=0.5cm]
	\tikzstyle{edge} = [thick]

    \node (r) [node, fill=skyblue, ultra thick] {};
    \node (l) [node, fill=white, above=0.25cm of r] {};

    \node (p1) [node, fill=orange, left=0.25cm of r] {};
    \node (p2) [node, fill=yellow, left=0.25cm of p1] {};
    \node (b1) [node, fill=bluegreen, left=0.25cm of p2] {};
    \node (b11) [node, fill=white, above left=0.25cm and 0.15cm of b1] {};
    \node (b12) [node, fill=white, below left=0.25cm and 0.15cm of b1] {};

    \node (b2) [node, fill=white, right=0.3cm of r] {};
    \node (b21) [node, fill=white, above right=0.25cm and 0.15cm of b2] {};
    \node (b22) [node, fill=white, right=0.3cm of b2] {};
    \node (b23) [node, fill=white, below right=0.25cm and 0.15cm of b2] {};

    \draw (r) edge [edge] (l);
    \draw (r) edge [edge] (p1);
    \draw (r) edge [edge] (b2);

    \draw (p1) edge [edge] (p2);
    \draw (p2) edge [edge] (b1);
    \draw (b1) edge [edge] (b11);
    \draw (b1) edge [edge] (b12);

    \draw (b2) edge [edge] (b21);
    \draw (b2) edge [edge] (b22);
    \draw (b2) edge [edge] (b23);

    \draw[ultra thick, rounded corners=3mm, densely dotted, color=black] ($(p1.north west) + (-0.15, 0.15)$) rectangle ($(r.south east) + (0.15, -0.15)$);
    \draw[ultra thick, rounded corners=3mm, loosely dashed, color=black] ($(b1.north west) + (-0.15, 0.25)$) rectangle ($(r.south east) + (0.25, -0.25)$);

    \draw[->, ultra thick, color=black, loosely dashed] ($(p2.south west) + (-0.3, -0.5)$) -> ($(p2.south west) + (-1, -1.3)$);
    \draw[->, ultra thick, color=black, densely dotted] ($(p1.south east) + (0.4, -0.5)$) -> ($(p1.south east) + (1.2, -1.3)$);
\end{tikzpicture}
        \vspace{0.1cm}

        \begin{tikzpicture}
    \tikzstyle{node} = [circle, draw, thick, minimum size=0.5cm]
	\tikzstyle{edge} = [thick]

    \node (b1) [node, ultra thick]  {};

    \node (l) [node, fill=white, above=0.25cm of r] {};

    \node (b11) [node, fill=white, above left=0.1cm and 0.3cm of b1] {};
    \node (b12) [node, fill=white, below left=0.1cm and 0.3cm of b1] {};

    \node (b2) [node, fill=white, right=0.3cm of b1] {};
    \node (b21) [node, fill=white, above right=0.25cm and 0.15cm of b2] {};
    \node (b22) [node, fill=white, right=0.3cm of b2] {};
    \node (b23) [node, fill=white, below right=0.25cm and 0.15cm of b2] {};

    \draw (b1) edge [edge] (b11);
    \draw (b1) edge [edge] (b12);
    \draw (b1) edge [edge] (l);
    \draw (b1) edge [edge] (b2);

    \draw (b2) edge [edge] (b21);
    \draw (b2) edge [edge] (b22);
    \draw (b2) edge [edge] (b23);

    \clip ($(b1)$)circle (0.2275cm);
    \begin{scope}
        \rotatebox{-45}{
            \fill[orange] ($(b1.north)$) rectangle ($(b1.west)$);
            \fill[skyblue] ($(b1.north)$) rectangle ($(b1.east)$);
            \fill[bluegreen] ($(b1.south)$) rectangle ($(b1.west)$);
            \fill[yellow] ($(b1.south)$) rectangle ($(b1.east)$);
        }
    \end{scope}
\end{tikzpicture}
        \hspace{1.8cm}
        \begin{tikzpicture}
    \tikzstyle{node} = [circle, draw, thick, minimum size=0.5cm]
	\tikzstyle{edge} = [thick]

    \node (p1) [node, ultra thick] {};
    \node (l) [node, fill=white, above=0.25cm of r] {};

    \node (p2) [node, fill=yellow, left=0.25cm of p1] {};
    \node (b1) [node, fill=bluegreen, left=0.25cm of p2]  {};
    \node (b11) [node, fill=white, above left=0.25cm and 0.15cm of b1] {};
    \node (b12) [node, fill=white, below left=0.25cm and 0.15cm of b1] {};

    \node (b2) [node, fill=white, right=0.3cm of p1] {};
    \node (b21) [node, fill=white, above right=0.25cm and 0.15cm of b2] {};
    \node (b22) [node, fill=white, right=0.3cm of b2] {};
    \node (b23) [node, fill=white, below right=0.25cm and 0.15cm of b2] {};

    \draw (p1) edge [edge] (l);
    \draw (p2) edge [edge] (b1);
    \draw (b1) edge [edge] (b11);
    \draw (b1) edge [edge] (b12);

    \draw (p1) edge [edge] (b2);
    \draw (b2) edge [edge] (b21);
    \draw (b2) edge [edge] (b22);
    \draw (b2) edge [edge] (b23);

    \clip ($(p1)$)circle (0.2275cm);
    \begin{scope}
        \fill[orange] ($(p1.south)$) rectangle ($(p1.north) + (-0.25, 0)$);
        \fill[skyblue] ($(p1.south)$) rectangle ($(p1.north) + (0.25, 0)$);
    \end{scope}

\end{tikzpicture}
    \end{center}
	\caption{
        Two possible ways (shown by dashed borders) for the district containing the current branch vertex (in bold) to overlap the segment shown by the colored vertices.
        Vertices are colored according to which candidate they vote for.
        The \wgm{} instances produced by our branching strategy are shown below, where vertices with multiple colors split votes between those candidates.
	}
    \label{fig:leaf_xp}
\end{figure}
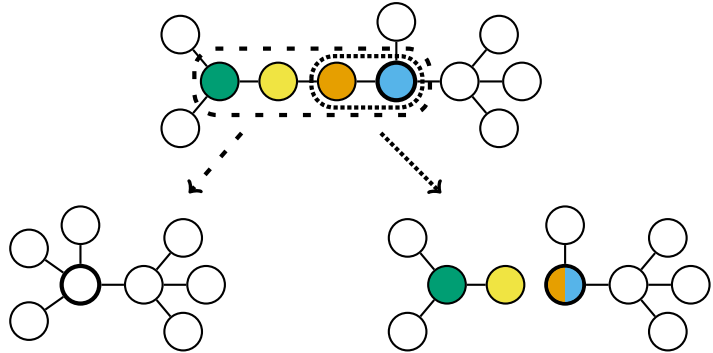

We now relate the parameter $\ell$ to the \textit{summed branch degree} which we define to be $d$, the sum of the degree of every vertex with degree $\delta(v) \geq 3$ (which we call \emph{branch vertices}).

\begin{lemma} \label{lem:leaf-to-branch}
    In a tree $T$ with $\ell$ leaves, the summed branch degree $d$ is at most $3\ell$.
\end{lemma}

\begin{proof}
    We induct on $n$, the number of vertices in $T$.
	For the base case, a tree with $n = 1$ vertex also has $\ell = 1$ and $d = 0$, so $d \leq 3\ell$ holds.
    Assume the claim holds for trees with less than $n$ vertices.
    Let $T$ be a tree with $n$ vertices, and let $u$ be a leaf in $T$ with neighbor $v$.
    Let $T'$ be the tree obtained by removing $u$ from $T$.
    By the inductive hypothesis, $d' \leq 3\ell'$ in $T'$.
    Consider the degree $\delta'$ of $v$ in $T'$.
    If $\delta'(v) = 1$, then $v$ is not a leaf in $T$, but $d = d'$ and $\ell = \ell'$.
    When $\delta'(v) = 2$, $d = d' + 3$ and $\ell = \ell' + 1$, so $d \leq 3\ell$ still holds.
    In the case that $\delta'(v) \geq 3$, $d = d' + 1$ and $\ell = \ell' + 1$, and so $d \leq 3\ell$ holds.
\end{proof}

Before describing the algorithm, we define a \emph{segment} of a tree $T$ to be any subpath of $T$ such that both endpoints are either a leaf or a branch vertex and all internal vertices have degree 2 in $T$.
The algorithm proceeds by selecting a branch vertex $b$ and a segment $S$ containing it.
It then branches on how the district $D$ containing $b$ could intersect $S$ (either ending at some vertex along $S$ or containing the entire segment).
The vertices in the same district as $b$ are contracted into $b$ and an edge is removed if $D$ does not contain all of $S$.
See Figure~\ref{fig:leaf_xp} for an example.

\begin{theorem} \label{thm:xp-alg}
    Let $(G, C, \chi, w, p, k)$ be an instance of \wgm{} with a tie-breaking rule $\eta$.
    If $G$ is a forest with summed branch degree at most $d$, then there exists an algorithm which can decide $(G, C, \chi, w, p, k)$ in time $O(n^d 2.619^k (n + m)^{O(1)})$.
\end{theorem}

\begin{proof}
    We proceed by induction on $d$.
    If $d = 0$, then $G$ is a path forest, and so the instance can be solved in $O(2.619^k (n+m)^{O(1)})$ time by Corollary~\ref{cor:disjoint_fpt_paths}.
    Assume the claim holds for graphs with summed branch degree less than $d$.

    Consider some branch vertex $b$ in $G$, and let $S = v_1, \dots, v_s$ be a segment of $G$ such that $b = v_1$.
    If $(G, C, \chi, w, p, k)$ is a YES-instance, then in any valid district-partition, there exists a district $D$ which contains $b$.
    Thus, either $D$ contains all of $S$ or there is a last vertex $v_i$ along $S$ which is still in $D$.
    We branch on the choice of $v_i$ and construct a new instance $(G', C, \chi, w', p, k)$ in the following manner.
    First if $i < s$, remove the edge $v_iv_{i+1}$.
    Then, contract $v_1, \dots, v_i$ into $b$ so that it is a single vertex with weight vector $w(b) = \sum_{j=1}^{i} w(v_j)$.
    Any district partition for $G'$ is easily converted to a district partition for $G$ by extending the district containing $b$ along $S$ to $v_i$.

    Finally, we argue the runtime is correct.
    The new instance $(G', C, \chi, w', p, k)$ has summed branch degree at most $d - 1$.
    Either the degree of $b$ is reduced by removing an edge and contracting part of $S$, or all of $S$ is contracted into a single vertex with degree $\delta(b) + \delta(v_s) - 2$.
    Thus, by the inductive hypothesis, the reduced instance can be solved in time $O(n^{d-1} 2.619^k (n + m)^{O(1)})$.
    Since $|S| \leq n$, we can check every branch in $O(n^d 2.619^k (n + m)^{O(1)})$ time.
    We note that contracting the edges to create $G'$ can be handled in constant time by checking the branches defined by $v_1, \dots, v_s$ in that order.
\end{proof}

\begin{corollary} \label{cor:xp-alg}
    Let $(G, C, \chi, w, p, k)$ be an instance of \wgm{} with a tie-breaking rule $\eta$.
    If $G$ is a tree with at most $\ell$ leaves, then there exists an algorithm which can answer $(G, C, \chi, w, p, k)$ using $\eta$ in time $O(n^{3\ell} 2.619^k (n + m)^{O(1)})$.
\end{corollary}

\begin{proof}
    This result follows from Theorem~\ref{thm:xp-alg} and Lemma~\ref{lem:leaf-to-branch}.
\end{proof}

\section{Conclusion}
Identifying and preventing political gerrymandering is an important social problem that has recently seen significant attention from the algorithmic community.
To incorporate socio-political relationships beyond geographic proximity, Ito et al.\ formalized \gm{} on graphs~\cite{ito2019}.
\gm{} is a natural candidate for FPT algorithms, since the number of districts $k$ is often manageably small in real-world instances (e.g.\ 10-15).
In contrast, XP algorithms are likely infeasible at these parameter values, and so the precise parameterized complexity of \gm{} has important practical consequences.

Ito et al.\ spurred interest in \gm{} on trees specifically by proving NP-completeness for the problem on $K_{2,n}$ (i.e.\ a graph one vertex deletion away from a tree)~\cite{ito2019}.
In response, \gm{} results have been discovered for many restricted settings including a polynomial time algorithm for stars~\cite{ito2019}, weak NP-hardness for trees with at least 3 candidates~\cite{bentert2023}, and an FPT algorithm for paths parameterized by $k$~\cite{gupta2021}.
We further characterize the properties of trees that make \gm{} hard.
First, we show that \uwgm{} is W[2]-hard with respect to $k$ in trees of depth 2, answering an open question of~\cite{gupta2021}.
Furthermore, we prove that \uwgm{} remains W[2]-hard with respect to $k + \ell$ in trees with only $\ell$ leaves, even if $G$ is a subdivided star (i.e.\ only has one vertex with degree greater than 2).
Complementing these results, we give an algorithm to solve \wgm{} that is FPT with respect to $k$ when $\ell$ is a fixed constant.
All together, this essentially resolves the parameterized complexity of \gm{} with respect to the number of districts.

\subsubsection*{Acknowledgements}
Thanks to Christopher Beatty for his contributions to a course project that led to this research.
We also thank the anonymous reviewers whose comments on a previous version of this manuscript led to significant improvements in notational clarity.

\newpage

\bibliographystyle{plain}
\bibliography{refs}

\end{document}